\documentclass[cis]{ipart_v1}

\Vol{17}
\Issue{4}
\Year{2017}
\firstpage{249}

\usepackage[english]{babel}

\usepackage{amsthm,amsmath,amssymb}
\newtheorem{thm}{Theorem}
\newtheorem{lemma}{Lemma}
\newtheorem{comment}{Comment}

\begin{document}

\title[PSPACE hardness of approximating the capacity]{PSPACE hardness of approximating the capacity of time-invariant Markov channels with perfect feedback}

\author[Mukul Agarwal]{Mukul Agarwal}

%\section{Acknowledgment}
%The author thanks Prof. John Tsitsiklis for insightful discussions. The author also thanks Dr. Tom Richardson and Prof. Sanjoy Mitter for comments on the writing of the paper and other technical discussions.

\begin{abstract}
It is proved that approximating the capacity of a time-invariant  Markov channel with perfect feedback is PSPACE hard.
\end{abstract}

\maketitle

\section{Introduction}

In this paper, it is proved that approximating the capacity of a Markov channel with perfect feedback is PSPACE-hard. By `approximating,' we mean, computing the capacity within a certain given additive error $e$. A class of channels will be constructed for which it will be proved that approximating the capacity to within $0.1$ bits of the correct capacity is PSPACE-hard.

The observation that approximating the capacity of a Markov channel with perfect feedback can be formulated as a stochastic dynamic programming problem with partial observations thereby connecting it to the work of Tsitsiklis and Papadimitriou on the complexity of Partially Observed Markov decision processes \cite{TsPaCo} is what is novel here along with the result.

The authors of \cite{TsPaCo} demonstrate that the complexity of Markov decision processes with partially observed states is PSPACE-hard. The constructions and arguments in our paper depend significantly on, and are in many cases, the same as the arguments in \cite{TsPaCo}. The output of the channel in this paper is the partial state information in \cite{TsPaCo}. By carefully chosing the reset probabilities from the final decision states, we are able to map the complexity result about Markov decision process in \cite{TsPaCo} to result concerning the hardness of capacity approximation in this paper. The utter simplicity with which a result in control theory can be `transported' into a result in communication theory can be seen.

For the relevant background on complexity theory, see \cite{TsPaCo} and references therein. An understanding of Section 4 in \cite{TsPaCo}, in particular, the statement of Theorem 6 and its proof is needed to understand the proof here. Once this is understood, and the reader has an understanding of information theory and basic probability theory, the proofs presented here can be understood.

\section{Literature survey}

In \cite{TsPaCo}, the complexity of Markov decision processes and partially observed Markov  decision processes has been considered and in \cite{TsPaIn}, the Witsenhausen and team decision problems have been considered. In both these papers, it is proved that there are problems in each category which are hard. See further references therein for problems which have been proved to be hard, especially in a control setting. There are various papers in information theory demostrating hardness of source and channel coding algorithms and code constructions, see, for example \cite{Wit}, \cite{Ber}, and \cite{Lang}. In \cite{Wit}, it has been proved that the generalized Lloyd-Max algorithm is NP-complete. In \cite{Ber}, it has been proved that general decoding problem for linear codes and the general problem of finding the weights of a linear code are both NP-complete. In \cite{Lang},  the problem of encoding complexity of network coding is considered, and one of the results in this paper is that approximating the minimum number of encoding nodes required  for the design of multicast coding networks is NP-hard.

In this paper, it is proved that even just approximately calculating the capacity to within a constant additive positive number is  PSPACE-hard for the problem of Markov channel with perfect feedback, and thus, for a general network, the problem of approximating the capacity region is PSPACE-hard.

\section{The result}

This section contains the channel construction, lemmas and the theorem.

For literature on Markov channels with feedback, the reader is referred to \cite{Tat} and references therein, though this paper is not particularly needed to understand our paper.

\subsection{Channel construction and reliable communication}

Starting from any quantified formula $Q \!=\! \exists x_1 \forall x_2 \exists x_3 \ldots \forall x_n F(x_1, x_2, \ldots, x_n)$ with $n$ variables and $m$ clauses $C_1, C_2, \ldots, C_m$, where $F$ is in conjunctive normal form, we construct a channel as follows:

Channel states $s_0, A_{ij}$, $A'_{ij}$, $T_{ij}$, $T'_{ij}$, $F_{ij}$, $F'_{ij}$; sets $A_j$, $A'_j$, $T_j$, $T'_j$, $F_j$, $F'_j$ $1 \leq i \leq m, 1 \leq j \leq n$ are the same as those in \cite{TsPaCo}. The states $A_{i,n+1}$ are lumped into a single state $A_{n+1}$ and the states $A'_{i,n+1}$ are lumped into a single state $A'_{n+1}$. Also, there is no terminating state, and from $A_{n+1}$, the channel goes back to itself or to $s_0$, and from $A'_{n+1}$, the channel goes back to itself or to $s_0$, in a way described later. 

The state transitions, which as stated above, are exactly the same as in \cite{TsPaCo} (but for the exception stated above). In order to make the paper self-contained, we reproduce the same here.

Based on each clause $C_i$ of $Q$ and variable $x_j$ there are $6$ states, $A_{ij},$ $A'_{ij},$ $T_{ij},$ $T'_{ij},$ $F_{ij},$ $F'_{ij}$. There are also $2$ additional states $A_{n+1}$ and $A'_{n+1}$.

The initial state $s_0$ should be thought of as a set. For each variable $j$, the states $A_{ij}, 1 \leq i \leq m$ form a set $A_i$, and the states $A'_{ij}, 1 \leq j \leq m$ form a set $A'_{ij}$ and similarly, the sets $T_i, T'_i$, $F_i, F'_i$ are defined (this will be the partial state information, as we shall see below).

The channel can transition from $s_0$ to states $A'_{i1}, 1 \leq i \leq m$ with equal probability. If $x_j$ is an existential variable, there are two possible state transitions out of the set $A_j$ leading with certainty from $A_{ij}$ to $T_{ij}$ and $F_{ij}$ respectively, and similarly, there are two transitions out of the set $A'_j$ leading with certainty from $A'_{ij}$ to $T'_{ij}$ and $F'_{ij}$ respectively. If $x_j$ is a universal variable, there is one transition  out of the set $A_j$ leading with equal probability from $A_{ij}$ to $T_{ij}$ and $F_{ij}$ and similarly, there is one transition out of the set $A'_j$ leading with equal probability from $A'_{ij}$ to $T'_{ij}$ and $F'_{ij}$. From the sets $T_j, F_j, T'_j, F'_j$ sets, there is only one transition which leads with certainty from $T_{ij}, F_{ij}, T'_{ij}$ and $F'_{ij}$ to (respectively) $A_{i,j+1}, A_{i, j+1}, A'_{i, j+1}, A'_{i,j+1}$ with two exceptions: If $x_j$ appears positively in $C_i$, the transition from $T'_{ij}$ is to $A_{i,j+1}$ instead of $A'_{i,j+1}$ and if $x_j$ appears negatively, the transition from $F'_{i,j}$ is to $A_{i,j+1}$. When the channel reaches the state $A_{i,n+1}$, as stated above, the states $A_{i,n+1}$ are lumped into a single state $A_{n+1}$, and the channel stays in $A_{n+1}$ with probability $1-p$ and transitions back to the state $s_0$ with probability $p$ for a value of $p$ stated later. Similarly, when the channel reaches state $A'_{i,n+1}$, as stated previously, the states $A'_{i,n+1}$ are lumped into a single state $A'_{n+1}$, and the channel stays at state $A'_{n+1}$ with probability $1-q$ and transitions to the state $s_0$ with probability $q$, for a value of $q$ stated later.

%Channel transitions are exactly as in \cite{TsPaCo}, except that at state $A_{i,n+1}$, the channel stays at state $A_{i,n+1}$ with probability $1-p$ and transitions to state $s_0$ with probability $p$. Similarly, from state $A'_{i,n+1}$, the channel either stays at the state $A'_{i,n+1}$ with probability $1-q$ or transitions to the state $s_0$ with probability $q$.

The state $A_{n+1}$ will be called `good' state and the rest of the states will be called `bad' states. The reason for this will become clear below.

The output of the channel is the partial state information, that is, one of $s_0, A_i,A'_i,T_i,T'_i,F_i,F'_i$, $1 \leq i \leq n$, and $A_{n+1}$, $A'_{n+1}$  along with an output bit either $0$ or $1$ depending on channel input and channel functioning.  Thus, the output space of the channel is $\mathbb O \triangleq \{s_0, \{T_i\}, \{T'_i\},\{F_i\}, \{F'_i\}\{A_{i}\}, \{A'_{i}\}, 1 \leq i \leq n, A_{n+1}, A'_{n+1}\} \} \times \{0,1\}$.

The output of the channel is `fed back' directly to the encoder \emph{without} delay. 

The input to the channel is $\{D_1, D_2\} \times \{0, 1\}$. Intuitively, this should be thought of as a bit being transmitted through $0$ or $1$ and $D_1$ and $D_2$ determine the state transition (note that there are at most 2 possible state transitions out of a state). Of course, another policy can be used instead of transmitting the information bit and the state transition information; the above is just an intuitive way of thinking about the channel input.

The input to the encoder is a sequence of bits, each bit taking a value of $1$ with probability $\frac{1}{2}$ and $0$ with probability $\frac{1}{2}$, along with the output of the channel, which, as stated above, is `fed back' directly to the encoder without delay. At each time, the bit input to the encoder should be thought of as a single bit, the bit which has not yet been communicated. Thus, the input space of the encoder is $\{0, 1\} \times \mathbb O$ $=$
$\{0, 1\} \times  \{s_0, \{T_i\}, \{T'_i\},\{F_i\}, \{F'_i\}\{A_{i}\},\allowbreak \{A'_{i}\}, 1 \leq i \leq n, A_{n+1}, A'_{n+1} \} \times \{0,1\}$.

Based on all past inputs (a bit stream and all past feedback from the channel output), the encoder makes an encoding and `feeds it' into the  channel.

It will be assumed that the channel starts in state $s_0$. Note that if the channel starts in state $s_0$, in $2n+1$ units of time, it reaches either state $A_{n+1}$ or state $A'_{n+1}$.
%Note the following: the functioning of channel state is sequential in the sense that it starts at $s_0$, at time 1, the channel is in some state $X_{i1}$ for some $i$, where $X$ stands for one of $A, A', T, T', F, F'$. At the next time point, the channel state is in some set $X_{i2}$. At time $2n+1$, the channel is in state $X_{i,n+1}$ where $X = A$ or $A'$, where the channel stays with a certain probability ($p$ or $q$) or jumps to $s_0$ and the process con tinues.

For the purpose of understanding, it is best to think of the problem as a sequential problem where a set of bits `enter' the encoder at a certain rate $R$ which causes an encoding  and the channel produces the outputs from which a decoding needs to happen. A rate $R$ is achievable if, for every $\delta$, $\exists t_{\delta}$ such that for $t > t_{\delta}$, $Rt$ bits can be communicated and the average error, that is, $\Pr(\hat B_{tR} \neq B_{tR})$ is less than $\delta$, where $B_{tR}$ denotes the bit input upto time $t$ and $\hat B_{tR}$ is the corresponding decoding.

Let $p=2^{-(mn)^{100}}$ and $q = 2^{-(mn)^{200}}$. %The way these are chosen satisfy at least the following conditions: $p,q \to 0$ as $mn \to \infty$ and $\frac{q}{p} \to 0$ at a sufficiently fast rate as $mn \to \infty$.%

\subsection{Lemmas, theorem and proofs}

As has been stated previously, assume, in what follows, that the channel starts in state $s_0$.

\begin{lemma} Given $\epsilon > 0$. Then, $\exists m_0, n_0$, depending only on $\epsilon$ such that for $m>m_0, n>n_0$, if $Q$ is true, capacity of the channel corresponding to $Q$ is larger than $1-\epsilon$.
\end{lemma}

\begin{proof}
  By \cite{TsPaCo}, $Q$ is true implies that we can choose the channel input (decisions in \cite{TsPaCo}) so that we always end up in $A_{n+1}$, not $A'_{n+1}$. The transitions from $s_0$ to $A_{n+1}$ takes $2n+1$ units of time where at worst, no bits can be communicated, and the channel stays in state $A_{n+1}$ for an average of order of magnitude $2^{(mn)^{100}}$ number of transitions, where $1$ bit is transmitted noiselessly per channel use. By these considerations, it follows that given any $\kappa > 0$, $\exists$ $m_1, n_1$ such that for $m>m_1$, $n>n_1$, the stationary distribution of the state $A_{n+1}$ of the Markov chain with the above chosen channel inputs is $>1-\kappa$. By use of the ergodic theory for Markov chains, the lemma follows. 
\end{proof}

\begin{lemma}
 Given $\alpha > 0$. Then, $\exists m_0, n_0$ sufficiently large, depending only on $\alpha$  such that for $m>m_0, n>n_0$, the capacity of the channel corresponding to the formula $Q$ is larger than $\alpha$ implies that $Q$ is true. 
\end{lemma}

\begin{proof}
If there was some way for the channel to enter the state $A'_{n+1}$ irrespective of the decisions, this would happen with probability at least $\frac{2^{-n}}{m}$ (see \cite{TsPaCo}), and then, the channel stays in this state for an average of an order of magnitude of $2^{(mn)^{200}}$ amount of time (1 unit of time refers to one state transition). Note that there is no transmission of information possible in state $A'_{n+1}$. Even if the channel ended in $A_{n+1}$ with the rest of the probability $1-\frac{2^{-n}}{m}$, the average order of magnitude amount of time the channel stays in state $A_{n+1}$ is $2^{(mn)^{100}}$ which is `much less' than $2^{(mn)^{200}}\frac{2^{-n}}{m}$. Also, there is the $2n+1$ units of time when the channel transitions from $s_0$ to $A_{n+1}$ or $A'_{n+1}$, which is `negligible'.  It follows that given any $\lambda > 0$, $\exists m_2, n_2$ such that for $m > m_2$, $n > n_2$, the fraction of time the channel spends in states $A'_{n+1}$ is $>(1-\lambda)$ with high probability. Finally, note that the amount of transmission of information during the $(2n+1)$ units of time when the channel transitions from $s_0$ to $A_{n+1}$ or $A'_{n+1}$ is at most $\log \lceil 6mn+3 \rceil$. This is because the output space to the channel has cardinality $6mn+3$.  By taking into account the above numbers, it follows, then, that if the channel could enter $A'_{n+1}$, there exist $m_0, n_0$ sufficiently large such that the capacity of the channel will be less than $\alpha$ which will contradict the assumption on the channel. It follows, then, that there is a set of decisions (channel inputs) for which the channel never enters the state $A'_{n+1}$ which implies, by \cite{TsPaCo}, that the formula is true. 
\end{proof}

\begin{thm} \label{TheoremMain}
Computing the capacity of this set of Markov channels (the set of channels formed by taking a channel corresponding to each formula $Q$ where $m,n$ and the particular formula can be any positive integers) when perfect feedback is available, to within an accuracy of 0.1 bits per channel use, is PSPACE hard.
\end{thm}

\begin{proof}
If the capacity of this set of Markov channels with feedback could be computed to within an accuracy of $0.1$, it would be known whether the capacity of the channel is less than $0.2$ or larger than $0.8$. This would imply, from the previous lemmas, that we would know, for sufficiently large $m,n$, whether $Q$ is true or not, and this problem is PSPACE hard.
\end{proof}

\section{Comments}

\begin{comment}
It has been assumed that the channel starts in state $s_0$. This is only for simplicity of presentation. Minor modifications can be made in order to make the channel start in any state.
\end{comment}

\begin{comment}
In addition, a transition from state $A_{i,n+1}$ to $A'_{i,n+1}$ with a probability $2^{-(mn)^{500}}$ could be added, and thus, from state $A_{i,n+1}$ to state $s_0$ with probability $1-2^{-(mn)^{500}} - 2^{-(mn)^{100}}$ to make the picture a little more realistic. Also, the bound in Lemma 2 concerning the  information transmission from input to output, is $\log \lceil 6nm + 3 \rceil$; however, a bound should also be possible on the cardinality of the input space of the channel.
\end{comment}

\begin{comment}
There is nothing special about the number $0.1$ in Theorem~\ref{TheoremMain}. Also, $p$ and $q$ need not be doubly exponential. Doubly exponentials work, and for that reason, they have been chosen this way.
\end{comment}

\begin{comment}
\emph{PSPACE} hardness implies \emph{NP} hardness and thus, the problem dealt with in this paper is also \emph{NP} hard.
\end{comment}

\begin{comment}
A block-coding model can be considered instead of a sequential model. For the purpose of intuition, it is best to think of a sequential model.
\end{comment}

\section{Importance of the result}

In information theory, one important research direction is to find single-letter characterizations, appropriately defined, for capacity regions of networks. However, it is not concretely known whether single letter characterizations exist in general. An example is the two-way channel \cite{S2way}. In order to prove that in general, there is no single letter characterization, all one needs is an example for which it is not possible to get one.

This paper takes a different view-point and firmly establishes a limit from the view-point of complexity theory by providing an example for which approximating the capacity of a general network is indeed hard in the sense of PSPACE-hardness.

\section{Recap and research directions}

A class of Markov channels with perfect feedback  was constructed for which it was proved that approximating the capacity to within 0.1 bits is PSPACE hard.

It would be worthwhile exploring the application of this proof idea to other channels with potentially noisy feedback, channels with perfect state information, and networks in general.

It would be helpful to see whether this result puts restrictions on the kind of single-letter characterizations there may exist for capacity regions of networks; for example, this paper will rule out certain characterizations which can be approximated in a way that is not PSPACE-hard.

The simplicity with which a result in which a result in control theory has been used to prove a result in information theory may be noted and further possibilities of the same may be explored.

\section{Acknowledgments}

The author thanks Prof. John Tsitsiklis for suggesting his paper \cite{TsPaCo} which led to proving the result in this paper. The author also thanks Dr. Tom Richardson for carefully reviewing this paper and suggesting extensive changes in the writing of the proof and encouraging the author to submit the paper. The author thanks Prof. Vincent Tan for helpful discussions. Finally, the author thanks Prof. Sanjoy Mitter for suggesting that the author talk to Prof. John Tsitsiklis and Dr. Tom Richardson, for looking through the first version of this paper, and for general encouragement.

%The author thanks Prof. John Tsitsiklis for insightful discussions. The author also thanks Dr. Tom Richardson and Prof. Sanjoy Mitter for comments on the writing of the paper and other technical discussions.

\address{Department of Electrical and Computer Engineering \\
Boston University, Boston, MA 02215, USA\\
\email{magar@alum.mit.edu}}

\end{document}